\documentclass[11pt]{article}
\usepackage{amsmath,amssymb,amsthm}
\usepackage[margin=1.15in]{geometry}
\usepackage{hyperref}
\usepackage{enumitem}
\usepackage{booktabs}

\newtheorem{theorem}{Theorem}

\newtheorem{corollary}{Corollary}
\newtheorem{remark}{Remark}

\title{On a Simple Relationship Between Order Imbalance,\\ Skew and Width in Over-The-Counter Trading}
\author{Peter Cotton\thanks{The work was completed around 2015, and the central symmetry was presented at NYU Tandon some time later; it was first written up in April 2022. The author thanks Pascal Tomecek, Morten Bai Andersen, Nick West and Erik Clossen for useful discussions. Joint work with Andrew Papanicolaou (\emph{Trading Illiquid Goods: Market Making as a Sequence of Sealed-Bid Auctions, with Analytic Results}, working paper) extends the model considered here to clustered arrivals and stochastically varying imbalance; the present paper isolates the original observation, which is a symmetry. Comments welcome.}}
\date{\today}

\begin{document}
\maketitle

\begin{abstract}
We consider a market maker who can only obtain and dispose of inventory by responding to a sequence of sealed-bid enquiries, and whose customers arrive with imbalanced intent: sellers more often than buyers, or the reverse. Under the assumption that the best competing response is exponentially distributed around a commonly discerned fair price, we observe a symmetry in the steady state solution that compresses the imbalanced problem onto the perfectly balanced one. Order imbalance is absorbed, exactly, by a translation of the market maker's skew, a widening of her quotes, and a multiplication of her effective cost of carry. The adjustment is simple even though the solution it adjusts is not, and it involves no free parameter beyond the observable market width. The exponential assumption is needed only locally, at the quotes actually made, and the width that enters is the locally observed one. Among the consequences: a market maker with zero inventory should still skew; skew responds to imbalance at first order whereas width responds only at second order; and the popular ``constant width, linear skew'' heuristic is recovered as the small-skew solution in the special case of balanced flow and quadratic holding cost.
\end{abstract}



\section{Introduction}

The facilitation of trade by intermediaries often takes the form of a sequence of sealed-bid auctions. A customer wishing to sell an item will typically obtain several bids from dealers and select the best price; a customer wishing to buy will perform a similar search and transact with whomever offers the lowest. The market maker takes on price risk and carrying cost, and must respond to each enquiry with that in mind. She is said to mark up her offer, or mark down her bid, relative to a fair market price.

We have in mind fungible goods for which this description is apt (dealer markets in physical commodities, corporate and municipal bonds, blocks of securities, or for that matter rare baseball cards) and one empirical regularity of such markets that the standard theory does not deliver. Dealers shade their quotes with the direction of the flow. A dealer who expects the next call to be a seller quotes differently from one who expects a buyer, and does so even when her book is flat. Practitioners regard this as too obvious to remark upon. The optimal market making literature has meanwhile concentrated on the other lever: the benchmark policies descending from Avellaneda and Stoikov \cite{avellaneda2008} hold width constant and skew linearly in inventory, with buy and sell arrivals assumed symmetric.

The purpose of this paper is to show that flow-shading is not a heuristic layered on top of inventory management but an exact feature of optimal behavior, and that it follows from a symmetry. In the steady state, the market making problem facing imbalanced flow compresses onto the perfectly balanced problem: the imbalance is absorbed, without approximation, by a translation of the skew, a widening of the quotes, and a multiplication of the effective cost of carry. The three adjustments are explicit constants involving nothing but the imbalance itself and the market width.

The solution being adjusted is not simple. The indifference cost that characterizes optimal behavior must in general be computed numerically, and its shape depends on the entire cost of carry function. The imbalance correction nevertheless rides on top of any such solution unchanged, so a single balanced solve covers every level of imbalance. Nor is the exponential assumption load-bearing in its global form: only the hazard at the quotes actually made enters, a point Remark~\ref{rem:local} makes precise.

We note three consequences. A market maker with zero inventory should still skew, by half the market width times the log-odds of the arrival direction, which is approximately twice the market width times the imbalance when the imbalance is mild. The width response is second order in the imbalance while the skew response is first order, which we believe is why practitioners skew before they widen. And the constant width, linear skew heuristic reappears in the special case of balanced flow and quadratic holding cost, which locates what that benchmark omits: the flow term. The magnitudes involve no free parameter beyond the observable width and are falsifiable in principle; we make no empirical claims here.

\subsection{The CWLS benchmark. Making constant width markets with skew linear in inventory.}\label{sec:cwlsintro}

To fix terminology, a dealer's \emph{width} is the difference between where she is prepared to bid and where she is prepared to offer, and her \emph{skew} is the difference between the average of the two and her estimate of fair value. The Constant Width Linear Skew heuristic maintains a fixed width and sets skew proportional to signed inventory, with a sign that leans the dealer towards a neutral position. Motivations for CWLS range from its prima facie reasonableness to formal justification as an approximate solution of a stochastic control problem \cite{avellaneda2008,glft2012}. These justifications require some approximation, however, and they assume symmetric flow throughout. From a practical perspective CWLS is not a straw man, because it is simple, and because it sometimes \emph{is} the solution. Section~\ref{sec:closed} makes the second point precise, and the symmetry of Section~\ref{sec:theorem} supplies the benchmark's missing column: what to do about imbalance.

\subsection{Literature}\label{sec:literature}

The dealer-inventory line begins with Garman \cite{garman1976}, whose dealer faces genuinely asymmetric Poisson buy and sell rates, the oldest imbalanced-flow model, but sets a single static price pair, the content being ruin rather than policy. Amihud and Mendelson \cite{amihud1980} are the closest classical antecedent in spirit: imbalanced arrivals move a preferred inventory position and monotone quotes, though the results are structural rather than closed form. Ho and Stoll \cite{hostoll1981} quote from value function differences, and the \emph{slope} half of the slope-and-convexity characterization of Section~\ref{sec:markup} is implicit there and in everything since; the \emph{convexity} half states that the discretionary component of a dealer's width reveals the second difference of her inventory cost exactly, so that her inventory cost can be read off her quoting. We have not found the latter stated elsewhere.

Avellaneda and Stoikov \cite{avellaneda2008} anchor the modern literature. Their constant width, linear skew policy is reached by dropping terms, and the dropped terms are precisely the ones needed to establish any rule regarding width. Constant width is an artifact of the approximation, not a conclusion of the model, and the line consequently has nothing to say about how width should respond to inventory or to flow.

The exact treatment of Gu\'eant, Lehalle and Fernandez-Tapia \cite{glft2012} takes arrival intensities to be symmetric; the one directional asymmetry handled in closed form in that line is drift in the mid-price, which shifts quotes but neither widens them nor taxes carry. Gu\'eant \cite{gueant2016,gueant2017} and Bergault et al.\ \cite{bergaultgueant2021,bergault2021closed} accommodate asymmetric intensities in generality but handle them numerically. We have found no equivalence statement, no log-odds shift, and no cost multiplier in that literature.

Cartea, Jaimungal and coauthors \cite{cjr2014,cj2016,cjp2015} treat order-flow imbalance as a predictive signal, a different object from the stationary structural imbalance considered here. Bergault and Gu\'eant \cite{bergaultgueant2023} observe numerically that flow-aware market makers skew their quotes even in the absence of inventory; Corollary~\ref{cor:zeroskew} is the closed form of that observation.

The request-for-quote literature \cite{fgp2017,gueantmanziuk2019,barzykin2023} arrived some years after this work was completed. Fermanian, Gu\'eant and Pu \cite{fgp2017} provide econometric support for the win-curve assumption we make, Hendershott and Madhavan \cite{hendershott2015} document the sealed-bid mechanism itself, and Butz and Oomen \cite{butzoomen2019} document electronic FX dealers skewing on flow at flat inventory, the practice that Corollary~\ref{cor:zeroskew} derives rather than assumes.

\section{Model for an imbalanced dealer market}\label{sec:model}

Trade opportunities arrive as a Poisson process with mean inter-arrival time $\tau$. Each arrival is a customer, met simultaneously by quotes from several market makers, and the trade is done at the winning quote. At each arrival there is a chance $q \in (0,1)$ that the customer is a seller (a \emph{buy opportunity} for our market maker) and chance $1-q$ that the customer is a buyer. All trades are of size $s$. There is a fair market price commonly discerned by all dealers; in the event of a buy opportunity our market maker submits a markdown $m^{\downarrow} > 0$ below it, and in the event of a sell opportunity a markup $m^{\uparrow} > 0$ above it. She wins if her response beats the best competing response, whose displacement beyond fair value has survival function $F$ with hazard rate $h(z)$; we will shortly take the hazard constant, an assumption that Remark~\ref{rem:local} weakens to local constancy, so that the best competitor is exponentially distributed with mean
\[
w = 1/h ,
\]
which we call the \emph{market width}. Every won trade concedes a fixed amount $\epsilon$ to adverse selection: when we mark up an item by $m$ we assume we make only $m - \epsilon$. Holding inventory $x$ (a real number; negative values accommodate short positions) costs $c(x)$ per unit time with $c(0) = 0$, and the game is played forever: we study the ergodic, steady state policy.

The model deliberately leaves two things out, and the companion paper extends it in exactly these directions: arrivals may cluster, and the imbalance $q$ may itself fluctuate. The symmetry is exact in the stationary core, which is what we isolate here.

\subsection{Indifference liquidation markup}\label{sec:nu}

Suppose there is a liquidation cost $\nu(x) \ge 0$ that the market maker can pay at any time to unload her inventory, and suppose she is indifferent between paying it and carrying on. The function $\nu$ can be thought of as the dealer-to-dealer ``haircut'' at which she would clear her book. It characterizes her policy, and two discrete derivatives of it will do all the work:
\[
S(x) = \frac{\nu(x+s) - \nu(x-s)}{2s}, \qquad
C(x) = \frac{\nu(x+s) - 2\nu(x) + \nu(x-s)}{2s},
\]
the slope and convexity of the indifference cost, per unit.

\subsection{Optimal markups}\label{sec:markup}

Define break-even markups, which we call \emph{strikes}: the response levels at which winning the enquiry leaves the market maker indifferent,
\[
K^{\downarrow}(x;s) = \epsilon + \frac{\nu(x+s)-\nu(x)}{s}, \qquad
K^{\uparrow}(x;s) = \epsilon + \frac{\nu(x-s)-\nu(x)}{s},
\]
for buying and selling respectively. (Note that $K^{\uparrow}$ may well be negative for positive inventory: a loaded dealer will pay to get out.) The best response solves $\sup_m\, (m - K)\big(1-F(m)\big)$, and evidently the optimal choice of markup lies precisely where the increase in profit exactly offsets the possibility of losing the trade. For say we increased the markup by $\Delta m$; then
\[
\overbrace{\big(m^{\uparrow} - K^{\uparrow}(x;s)\big)}^{existing\ benefit}\;
\overbrace{h(m^{\uparrow})\,\Delta m}^{chance\ of\ losing\ it}
\;=\;
\overbrace{\Delta m}^{increase\ in\ profit} .
\]
The first-order condition is therefore independent of the distributional assumption on the inside market:
\begin{equation}
m^{\uparrow}(x;s) \;=\; \max\!\Big( \frac{1}{h\big(m^{\uparrow}(x;s)\big)} + K^{\uparrow}(x;s),\; 0 \Big),
\label{eq:foc}
\end{equation}
and likewise for $m^{\downarrow}$ with $K^{\downarrow}$. Under the constant hazard assumption this reads straightforwardly:
\begin{center}
``markup \;=\; market width \;+\; adverse selection \;+\; marginal inventory cost''
\end{center}
though of course the coupled quantities $\nu$ and the markups remain to be found. The first order condition itself is Lerner logic with survival-function demand and is classical; the market making content lies in the indifference-adjusted strike. Away from the region where the $\max(\cdot,0)$ binds, the two responses arrange themselves around a displaced midpoint:
\begin{equation}
\text{half-width}(x) = \Delta + C(x), \qquad
\text{mid}(x) = \text{fair value} - S(x), \qquad \Delta = w + \epsilon .
\label{eq:quotes}
\end{equation}
Slope of the inventory cost is the skew; convexity of the inventory cost is the discretionary width. These identities are exact, and the second runs both ways: a dealer's quoting reveals the first two derivatives of her inventory cost, whether or not she has ever written it down.

\subsection{The consistency equation}\label{sec:bellman}

The market maker can remain indifferent in two ways: liquidate now and re-enter, or carry to the next opportunity. Comparing the two yields a functional equation for $\nu$. With exponential markups, $\sup_m (m-K)e^{-hm} = \tfrac{1}{h}e^{-1-hK}$, and the steady state condition takes the form: for all $x$,
\begin{equation}
\frac{\tau\, c(x)}{s}
\;=\; \frac{e^{-1-h\epsilon}}{h}\,
\Big[\, e^{-hC(x)}\Big(\,q\,e^{-hS(x)} + (1-q)\,e^{hS(x)}\Big) \;-\; \big(\text{same expression at } x=0\big) \Big].
\label{eq:bellman}
\end{equation}
In words: the direct cost of carrying $x$ must be exactly offset by the differential option value of the next trading opportunity at $x$ versus at zero inventory. This is why common-sense holding cost arguments, mean holding time multiplied by carrying cost and so forth, systematically overstate the cost of inventory and encourage overly defensive quoting. They ignore that a larger inventory makes the next enquiry \emph{more} valuable, since the potential to get out helps more.

\section{The symmetry}\label{sec:theorem}

The paper rests on an elementary identity: for any $S$,
\begin{equation}
q\,e^{-hS} + (1-q)\,e^{hS}
\;=\; 2\sqrt{q(1-q)}\;\cosh\!\big(h(S - \delta)\big),
\qquad
\delta = \frac{1}{2h}\,\log\frac{q}{1-q} .
\label{eq:identity}
\end{equation}
An imbalanced mixture of exponentials is a balanced one in disguise: translated by $\delta$, and rescaled by $2\sqrt{q(1-q)} \le 1$. Writing $\gamma = \tfrac{1}{h}\log\tfrac{1}{2\sqrt{q(1-q)}} \ge 0$ for the rescaling expressed as a width, we have immediately:

\begin{theorem}[The imbalanced problem compresses onto the balanced one]\label{thm:main}
The steady state market making problem with arrival imbalance $q$ and carrying cost $c(\cdot)$ has the same solution as the perfectly balanced problem ($q = \tfrac12$) with carrying cost $M(q)\,c(\cdot)$, where
\begin{equation}
M(q) \;=\; e^{h\gamma} \;=\; \frac{1}{2\sqrt{q\,(1-q)}} \;\ge\; 1,
\label{eq:multiplier}
\end{equation}
under the correspondence: skew translated by $\delta$, and non-discretionary width $\Delta$ widened to $\Delta + \gamma$.
\end{theorem}

\begin{proof}
Substitute \eqref{eq:identity} into \eqref{eq:bellman} and write $S_\delta(x) = S(x) - \delta$:
\[
\frac{\tau\,c(x)}{s} = \frac{e^{-1-h\epsilon-h\gamma}}{h}
\Big[ e^{-hC(x)}\cosh\big(hS_\delta(x)\big) - e^{-hC(0)}\cosh\big(hS_\delta(0)\big) \Big],
\]
which is the balanced consistency equation in the pair $(C, S_\delta)$, with the cost side inflated by $e^{h\gamma} = M(q)$.
\end{proof}

\begin{remark}[The adjustment is simple; the solution is not]
The theorem neither requires nor provides a closed form. In general $\nu$ must be computed numerically, and everything about its shape depends on $c$. What the symmetry provides is the correction: whatever the balanced solution is, the imbalanced one is that same object translated by $\delta$, widened by $\gamma$, and evaluated at carrying cost $M(q)\,c$. One balanced solve serves every imbalance.
\end{remark}

\begin{remark}[Exponentiality is needed only locally]\label{rem:local}
The constant hazard need not hold globally. The distribution of the best competing quote enters through one object only: the value of the next enquiry as a function of its strike, $G(K) = \sup_{m}\,(m-K)\big(1-F(m)\big)$. For any smooth $F$ the envelope theorem gives $G'(K) = -(1-F(m^{*}(K)))$, and dividing by $G$,
\[
-\frac{d}{dK}\,\log G(K) \;=\; \frac{1}{m^{*}(K)-K} \;=\; h\big(m^{*}(K)\big),
\]
exactly. The proof uses only the exponential form of $G$, so the requirement is that the hazard be constant at the quotes actually made, across the bounded range of strikes generated by the inventories the dealer visits. If the hazard varies slowly there, every statement above holds with $w$ the reciprocal hazard at the operating point, and the error is of the order of the hazard's relative variation across that range. We checked this numerically in \texttt{verify\_local\_exponentiality.py}, which accompanies the paper: the error in the zero-inventory skew is about a quarter of the hazard's relative variation across the visited strikes, and shrinks linearly with it.
\end{remark}

\begin{remark}[The algebra is old]
We arrived at the substitution in the proof independently and noticed its provenance only afterwards. It is a birth--death symmetrization, and the geometric mean of the rates it produces is familiar from the spectral theory of birth and death processes \cite{ledermann1954,karlinmcgregor1957} and from the transient analysis of the simple queue \cite{bailey1954}. It could equally be applied to the linear systems of the limit order book literature with asymmetric intensities, though to our knowledge nobody has done so. What we claim is the economics: imbalance translates the skew, widens the quotes, and taxes the carry, with all three effects tied to a single observable.
\end{remark}

\begin{corollary}[A flat book still skews]\label{cor:zeroskew}
By symmetry the balanced solution has $S_\delta(0) = 0$, hence
\begin{equation}
S(0) \;=\; \delta \;=\; \frac{w}{2}\,\log\frac{q}{1-q} \;\approx\; 2\,w\,\Big(q - \tfrac12\Big)
\quad\text{for mild imbalance.}
\label{eq:zeroskew}
\end{equation}
When customers are predominantly sellers the market maker marks her midpoint down by $\delta$ before holding any inventory at all. Skew is a flow phenomenon first and an inventory phenomenon second.
\end{corollary}

\begin{corollary}[Skew is first order in imbalance; width is second order]\label{cor:orders}
Expanding about $q = \tfrac12$,
\[
\delta \approx 2w\Big(q-\tfrac12\Big), \qquad
\gamma \approx 2w\Big(q-\tfrac12\Big)^{2}, \qquad
M(q) \approx 1 + 2\Big(q-\tfrac12\Big)^{2}.
\]
A mild imbalance moves the midpoint linearly and the width only quadratically. This accords with the practitioner's instinct to skew into flow and widen only when the market becomes seriously one-sided. Both responses involve no parameter beyond the observable width $w$; the skew flips sign between net-selling and net-buying regimes while the widening, being even in $q - \tfrac12$, does not.
\end{corollary}

\begin{corollary}[Imbalance acts as a carrying cost]\label{cor:carry}
By \eqref{eq:multiplier}, one-sided flow is indistinguishable from a multiplied cost of carry, without bound as $q \to 0$ or $1$. A market maker in a good that costs nothing to hold is nonetheless taxed on inventory whenever the flow is one-sided, and taxed at a rate she can compute from the tape.
\end{corollary}

\section{Constant width solutions, and where CWLS sits}\label{sec:closed}

Suppose the convexity is constant, $C(x) \equiv C_0$: a constant width market. The consistency equation then inverts to give the skew directly,
\begin{equation}
S_\delta(x) \;=\; \frac1h \cosh^{-1}\!\Big( e^{hC_0}\,\Omega\,c(x) + \cosh\big(hS_\delta(0)\big) \Big),
\qquad \Omega = \tfrac{\tau h}{s}\,e^{1+h\epsilon+h\gamma},
\label{eq:coshinv}
\end{equation}
an easily implementable constant width model with, in general, \emph{nonlinear} skew. For small argument the right side is approximately proportional to $\sqrt{c(x)}$, so a linear component in $c$ (a physical holding fee) puts a square-root kink in the skew at the origin, and funding costs or equity hurdles bend it elsewhere.

Where, then, is CWLS exactly optimal? Requiring $C(x) \equiv C_0$ together with consistency of \eqref{eq:bellman} and non-negative liquidation costs forces
\begin{equation}
\nu(x) \;=\; \frac{C_0}{s}\,x^{2} \;+\; \delta\,x \;+\; C_2 ,
\label{eq:nuquad}
\end{equation}
with the imbalance appearing as the \emph{linear term of the indifference cost}, and feeding \eqref{eq:nuquad} back through \eqref{eq:bellman} pins the cost of carry to the shape $c(x) \propto \cosh(2hC_0 x/s) - 1$. The quadratic term prices inventory risk; the linear term prices the flow; and the skew $S(x) = 2\tfrac{C_0}{s}\,x + \delta$ is linear with an intercept.

CWLS is therefore exactly correct in one corner: balanced flow and a cosh-shaped cost of carry, a shape that is quadratic up to a relative error of $(hS_\delta(x))^{2}/12$. For a genuinely quadratic cost, CWLS is accurate while the skew is small relative to the width, with error quadratic in that ratio. The symmetry supplies the benchmark's missing column: under imbalance, keep the linear response to inventory, add the intercept $\delta$, and widen by $\gamma$. The moral for anyone outside the corner is procedural: solve for $\nu(x)$ rather than pretend $\nu \propto x^2$ works, since the quadratic ansatz is internally inconsistent for any other cost of carry.

\section{Uses}\label{sec:uses}

The equivalence was stated as a fact about optimal policy, but its practical content is broader, and we note five uses.

\emph{Fill ratios.} At the optimum, profitability factors cleanly. By \eqref{eq:foc} the optimal response sits the distance $w$ beyond its strike, so the expected net gain on a won trade is $s\,w$ at every level of inventory, and all variation in expected profit per enquiry is variation in the fill ratio. Under the constant hazard the fill ratio on a quote at markup $m$ is $e^{-m/w}$, log-linear in the markup with slope $-1/w$, which is how $w$ is estimated from a quote archive; at the optimal response to strike $K$ it equals $e^{-1-K/w}$. In particular a dealer facing zero adverse selection and zero marginal inventory cost should, if she quotes optimally, win $e^{-1} \approx 37\%$ of her enquiries, and persistent departures from that ratio measure her effective strike.

\emph{Testing the symmetry.} Request-for-quote archives record direction, quotes, and fills. Direction counts estimate $q$, the win curve estimates $w$, with Remark~\ref{rem:local} indicating that the local estimate is the appropriate one, and the theorem then predicts the zero-inventory skew $\delta = \tfrac{w}{2}\log\tfrac{q}{1-q}$ and the widening $\gamma$ with no free parameter. Cross-sectionally, instruments with more one-sided flow should show mids shifted at first order in the imbalance and quotes widened at second order, per Corollary~\ref{cor:orders}.

\emph{Reading a dealer's book.} The identities \eqref{eq:quotes} run both ways: skew and discretionary width are the first two differences of the indifference cost, so a counterparty's quoting reveals her inventory cost function, and a shift in her zero-inventory skew reveals her view of the flow.

\emph{Increasing the statistical power of learned policies.} Reinforcement learning approaches to market making \cite{gueantmanziuk2019} generalize the present model and need not assume an exponential win curve. The symmetry helps them anyway, reinterpreted as an approximate heuristic: experience gathered at one imbalance transfers to every other through the correction $(\delta, \gamma, M(q))$, so a learner can pool episodes across flow regimes instead of learning each regime separately. An invariance of this sort reduces the effective dimension of the policy space, and the sample sizes required shrink accordingly.

\emph{Assessing decisions after the fact.} Illiquidity implies short time series, and short time series make back-testing unreliable; the model instead prices individual decisions. Comparing a trader's response $m$ against the model's $m^{*}$ on a single enquiry:

\begin{center}
\begin{tabular}{@{}lcc@{}}
\toprule
 & model trades & model misses \\
\midrule
trader trades & $s\,(m - m^{*})$ & $s\,(m - m^{*}) + s\,w$ \\
trader misses & $-\,s\,w$ & $0$ \\
\bottomrule
\end{tabular}
\end{center}

\noindent Every cell is computable from observables at the moment of the enquiry, and the running total is a score of quoting skill that separates flow-reading (the $\delta$ term) from inventory management (the $S_\delta$ term). Profit and loss, by contrast, mixes both with luck. The constant margin under exponential demand is itself a classical property; the diagnostic built on it is not.


\begin{thebibliography}{19}
\bibitem{garman1976} Garman, M.~B. (1976). Market microstructure. \emph{Journal of Financial Economics} 3(3), 257--275.
\bibitem{amihud1980} Amihud, Y., Mendelson, H. (1980). Dealership market: market-making with inventory. \emph{Journal of Financial Economics} 8(1), 31--53.
\bibitem{hostoll1981} Ho, T., Stoll, H.~R. (1981). Optimal dealer pricing under transactions and return uncertainty. \emph{Journal of Financial Economics} 9(1), 47--73.
\bibitem{avellaneda2008} Avellaneda, M., Stoikov, S. (2008). High-frequency trading in a limit order book. \emph{Quantitative Finance} 8(3), 217--224.
\bibitem{glft2012} Gu\'eant, O., Lehalle, C.-A., Fernandez-Tapia, J. (2012). Dealing with the inventory risk: a solution to the market making problem. \emph{Mathematics and Financial Economics} 7(4), 477--507.
\bibitem{gueant2016} Gu\'eant, O. (2016). \emph{The Financial Mathematics of Market Liquidity: From Optimal Execution to Market Making}. Chapman \& Hall/CRC.
\bibitem{gueant2017} Gu\'eant, O. (2017). Optimal market making. \emph{Applied Mathematical Finance} 24(2), 112--154.
\bibitem{cjr2014} Cartea, \'A., Jaimungal, S., Ricci, J. (2014). Buy low, sell high: a high frequency trading perspective. \emph{SIAM Journal on Financial Mathematics} 5(1), 415--444.
\bibitem{cj2016} Cartea, \'A., Jaimungal, S. (2016). Incorporating order-flow into optimal execution. \emph{Mathematics and Financial Economics} 10(3), 339--364.
\bibitem{cjp2015} Cartea, \'A., Jaimungal, S., Penalva, J. (2015). \emph{Algorithmic and High-Frequency Trading}. Cambridge University Press.
\bibitem{fgp2017} Fermanian, J.-D., Gu\'eant, O., Pu, J. (2017). The behavior of dealers and clients on the European corporate bond market: the case of multi-dealer-to-client platforms. \emph{Market Microstructure and Liquidity} 2(3--4), 1750004.
\bibitem{gueantmanziuk2019} Gu\'eant, O., Manziuk, I. (2019). Deep reinforcement learning for market making in corporate bonds: beating the curse of dimensionality. \emph{Applied Mathematical Finance} 26(5), 387--452.
\bibitem{bergaultgueant2021} Bergault, P., Gu\'eant, O. (2021). Size matters for OTC market makers: general results and dimensionality reduction techniques. \emph{Mathematical Finance} 31(1), 279--322.
\bibitem{bergault2021closed} Bergault, P., Evangelista, D., Gu\'eant, O., Vieira, D. (2021). Closed-form approximations in multi-asset market making. \emph{Applied Mathematical Finance} 28(2), 101--142.
\bibitem{barzykin2023} Barzykin, A., Bergault, P., Gu\'eant, O. (2023). Algorithmic market making in dealer markets with hedging and market impact. \emph{Mathematical Finance} 33(1), 41--79.
\bibitem{hendershott2015} Hendershott, T., Madhavan, A. (2015). Click or call? Auction versus search in the over-the-counter market. \emph{Journal of Finance} 70(1), 419--447.
\bibitem{butzoomen2019} Butz, M., Oomen, R. (2019). Internalisation by electronic FX spot dealers. \emph{Quantitative Finance} 19(1), 35--56.
\bibitem{bergaultgueant2023} Bergault, P., Gu\'eant, O. (2023). Liquidity dynamics in RFQ markets and impact on pricing. arXiv:2309.04216.
\bibitem{ledermann1954} Ledermann, W., Reuter, G.~E.~H. (1954). Spectral theory for the differential equations of simple birth and death processes. \emph{Philosophical Transactions of the Royal Society A} 246(914), 321--369.
\bibitem{karlinmcgregor1957} Karlin, S., McGregor, J. (1957). The differential equations of birth-and-death processes, and the Stieltjes moment problem. \emph{Transactions of the American Mathematical Society} 85(2), 489--546.
\bibitem{bailey1954} Bailey, N.~T.~J. (1954). A continuous time treatment of a simple queue using generating functions. \emph{Journal of the Royal Statistical Society, Series B} 16(2), 288--291.
\end{thebibliography}
\end{document}